\documentclass[11pt,letterpaper]{article}
\usepackage{fullpage,cite}
\usepackage{amsmath,amsthm,amsfonts,commath,xcolor,graphicx,hyperref}

\usepackage[ruled,linesnumbered]{algorithm2e}
\usepackage{algpseudocode}

\newtheorem{theorem}{Theorem}[section]

\newtheorem{lemma}[theorem]{Lemma}

\newtheorem{definition}[theorem]{Definition}

\newcommand{\eat}[1]{}

\newcommand{\OPT}{\mathrm{OPT}}

\newcommand{\POPT}{\mathrm{OPT}_\mathrm{prod}}
\newcommand{\LOPT}{\mathrm{OPT}_\mathrm{log}}
\newcommand{\DUAL}{\mathrm{DUAL}}
\newcommand{\ICA}{\mathsf{ICA}}

\newcommand{\rmax}{\rho_{\max}}

\newcommand{\umaxi}{u_i^{\max}}

\newcommand{\pub}{\mathrm{pub}}

\newcommand{\ICACP}{\mathsf{ICA\text{-}CP}}
\newcommand{\ICAD}{\mathsf{ICA\text{-}D}}

\title{Approximations for Allocating Indivisible Items with Concave-Additive Valuations}
\author{
Nathaniel Kell\thanks{Email: {\tt kelln@denison.edu}}\\Denison University
\and
Kevin Sun\thanks{Email: {\tt ksun@cs.duke.edu}}\\Duke University
}

\begin{document}
\maketitle
\begin{abstract} 
We study a general allocation setting where agent valuations are {\em concave additive}. In this model, a collection of items must  be uniquely distributed among a set of agents, where each agent-item pair has a specified utility. The objective is to maximize the sum of agent valuations, each of which is an arbitrary non-decreasing concave function of the agent's total additive utility.  This setting was studied by Devanur and Jain (STOC 2012) in the online setting for divisible items. In this paper, 
we obtain {\em both  tight multiplicative and additive approximations in the offline setting for  indivisible items.}
Our approximations depend on novel parameters that measure the local multiplicative/additive curvatures of each agent valuation, which we show correspond directly to the integrality gap of the natural assignment convex program of the problem. Furthermore, we extend our additive guarantees to obtain {\em constant multiplicative approximations for Asymmetric Nash Welfare Maximization} when agents have  {\em smooth valuations} (Fain et al.\ EC 2018, Fluschnik et al.\ AAAI 2019). This algorithm also yields an interesting tatonnement-style interpretation,  where agents adjust uniform prices and items are assigned according to maximum {\em weighted} bang-per-buck ratios.
\end{abstract}

\section{Introduction}
\label{sec:intro}

The study of {\em indivisible} allocation has received increasing attention in recent years: given a collection of indivisible items and a set of agents --- each with a specified valuation function --- how should items be uniquely distributed among the agents as to maximize a specified measure of overall welfare? 
Many classic allocation and market models consider divisible goods that can be split fractionally among agents. However in the context of optimization,  there are many methods for computing optimal fractional assignments in polynomial time (e.g. ellipsoid methods), while indivisible variants are often NP-hard and  more algorithmically challenging.  
The consideration of other important model features and quality measures adds further complexity to the indivisible setting, such as {\em diminishing returns} in agent valuations and {\em ensuring fairness}. 
How do we allocate items uniquely when agents value additional items less on the margins? 
How do we find allocations that are efficient with respect to utilitarian welfare, while not completely starving some subsets of the agents? 
The focus of much recent work has been aimed at addressing such questions \cite{anari2018nash, barman2018finding, chakrabarty2010approximability, garg2018approximating, kalaitzis2015configuration,li2021constant}.

In this paper, we study a general yet natural allocation model that lies at the intersection of these algorithmic challenges, which we call {\em Indivisible Allocation with Concave-Additive Valuations} ($ICA$).
As input, we are given a set of $n$ agents and $m$ indivisible items, where each agent $i$ has a specified utility $u_{i,j}$ for an item $j$. An algorithm must partition the items into $n$ disjoint sets $(A_1, A_2, \ldots A_n$), one for each agent. The overall valuation agent $i$ has for her set is  $v_i(u_i)$, where her valuation $v_i: {\mathbf R}_{+} \rightarrow {\mathbf R}$ is permitted to be any monotone (non-decreasing) concave function of her total additive utility $u_i = \sum_{j \in A_i} u_{i,j}$. The objective is to maximize the utilitarian welfare of the allocation, i.e., the sum of agent valuations $\sum_iv_i( u_{i})$.  For ease of comparison to prior work, we refer to such valuations $v_i(\cdot)$ as {\em concave additive}.

\vspace{3mm} 
\noindent {\em Motivation and Background.} $\ICA$ is an indivisible variant of the fractional online problem studied by Devanur and Jain \cite{devanur2012online}. This model was primarily motivated by applications in internet advertising, where agents correspond to advertisers and items correspond to so-called ``impressions" (opportunities to show ads to users). In this setting, $u_{i,j}$ would  translate to the bid value an advertiser $i$ is willing pay for impression $j$.  The concave objective is then used to capture common contract features in internet ad systems such as under-delivery penalties and soft budgets. 

However, given the special role concave functions have played in the economics literature, we believe $\ICA$ is a natural problem to consider in its own right. 
For example, a standard class of valuations is that of {\em separable concave valuations} \cite{anari2018nash, chaudhury2018fair, chen2009settling, vazirani2011market}, which can be decomposed into the sum of monotone concave functions over the amount received from each good (or in the indivisible settings, the numbers of {\em copies} of a single good).
Thus, such valuations can  express varying degrees of diminishing returns for receiving more of the same item. 

But in many applications, agent valuations are in fact {\em inseparable}, since how much an agent values an additional item will likely depend on the other items she has already received. 
A canonical example of inseparable valuations are {\em budget-additive} functions, i.e., each valuation $v_i(u_i) = \min\left(u_i, c_i\right)$ where $c_i$ denotes the utility cap for agent $i$. 
One can view concave-additive valuations as a general class of inseparable functions that can express diminishing returns beyond simply a global cap on an agent's overall utility. 
There is a line work that has extensively studied approximation algorithms for welfare maximization of budget-additive valuations \cite{garg2001approximation, andelman2004auctions, azar2008improved, chakrabarty2010approximability, kalaitzis2015configuration}. To the best of our knowledge, the approximability of the concave-additive case has yet to be considered for indivisible items. 

\vspace{3mm} 
\noindent {\em Extension to Nash Welfare Maximization.} Another key motivation for considering concave-additive functions is that an {\em additive} approximation for $\ICA$
translates to a standard multiplicative approximation for the problem of {\em Nash Welfare Maximization}. 
Here, the objective is to maximize the weighed product of valuations $\left(\prod_i (v_i(u_i))^{\eta_i}\right)^{1/\eta}$, where $\eta_i > 0$ is the weight of each agent and $\eta = \sum_i \eta_i$ is the sum of weights.  
The main appeal of Nash welfare is that it naturally strikes a balance between optimizing utilitarian welfare versus max-min fairness (e.g., completely starving even just one agent will result in an overall objective of zero).

Observe that we can convert the objective to an instance of $\ICA$ by taking the logarithm, giving us the objective $\frac{1}{\eta}\sum_i \eta_i\ln(v_i(u_i))$. As long as each valuation $v_i(\cdot)$ is monotone and concave, we obtain an instance of $\ICA$ since the logarithm of a monotone concave function is still monotone and concave. Furthermore, an additive approximation of $\alpha$ for this log objective translates to a multiplicative approximation of $e^{\alpha}$ for the original product objective.

It has been known for over eighty years that the optimal Nash-welfare assignment for divisible goods can be obtained by solving the famous Eisenberg-Gale convex program \cite{eisenberg1959consensus}. 
However, the indivisible case was less-understood until the recent work of Cole and Gkatzelis~\cite{cole2015approximating}, who obtained the first constant-approximation algorithm for {\em symmetric agents} with {\em additive valuations} (i.e., for all agents $i$ we have $\eta_i =1$ and  $v_i(u_i) = u_i$).
Since this result, the problem has been heavily studied for indivisible items, both in terms of its approximability and its appealing fairness properties. 
Much of the work on approximations has focused on the symmetric  case, which  has produced a line of results obtaining constant approximations for subsequently more general valuation functions, such as budget additive \cite{garg2018approximating}, separable piecewise-linear \cite{anari2018nash}, Rado valuations \cite{garg2021approximating}, and then recently for general submodular valuations \cite{li2021constant}.

Unfortunately even for additive valuations, the case of asymmetric agents (general weights $\eta_i > 0$) still remains a key open problem. The best known approximation for additive valuations is currently $O(n)$ by a result of Garg et al.~\cite{garg2001approximation}, which also extends to budget-additive functions. Very recently, polynomial-time approximation schemes were given by Garg et al.~\cite{garg2021tractable} for the special cases of identical agents ($u_{i,j} = u_j$ for all $i$) and two-value agents (for all agents $i$, $u_{i,j} > 0$ for at most two items $j$).

\subsection{Contributions and Techniques} 
In this paper, we obtain {\em tight multiplicative and additive approximations for the $ICA$ problem}.
Our approximations depend on novel curvature parameters that measure the degree of local change that can occur in the concave valuation of each agent, which may be of independent interest. 
For a given agent valuation  $v_i(\cdot)$ and maximum item utility  $\max_{j}u_{i,j}$ for agent $i$, we denote the {\em local multiplicative and additive curvatures} of $v_i(\cdot)$ as $\mu_i$ and $\alpha_i$, respectively, where higher/lower values of $\mu_i$ and $\alpha_i$ correspond to higher/lower degrees of local curvature  in $v_i(\cdot)$ with respect to a local change of $\max_{j}u_{i,j}$.
These parameters are formally defined and diagrammed in our preliminaries section (Section \ref{subsec:lcb}). 

Furthermore, we show that our additive guarantees for $\ICA$ extend to {\em new approximations for Asymmetric Nash Welfare with smoothed agent valuations}. The details and applications of each approximation type are stated formally below.

\subsubsection{Multiplicative Guarantees} 
\label{subsub:mult-guar}
For our main result, we give a  polynomial-time algorithm for any general instance of $\ICA$ that achieves an approximation ratio of $(1+\epsilon)\max_{i}\mu_i$.
Note that in the setting of multiplicative guarantees, we assume that each agent valuation is non-negative ($v_i:{\mathbf R}_+\rightarrow{\mathbf R}_+$) and differentiable.

\begin{theorem} 
\label{thm:pd-alg}
Consider an instance of $\ICA$ such that each $v_i:{\mathbf R}_+ \rightarrow {\mathbf R}_+$ has well-defined first derivative.  
Let $\umaxi = \sum_j u_{ij}$, and let $\rmax = \max_i(v_i'(0)\umaxi/v_i(\umaxi))$. Then there exists an algorithm that achieves an approximation of  $(1+\epsilon)\max_{i}\mu_i$ for $\ICA$ that runs in time $O(mn\ln(\rmax/\epsilon)/\epsilon)$.
\end{theorem}

We emphasize that our approach is efficient, easy to implement, and does not require  convex program solvers. Instead, it leverages the duality techniques introduced  in \cite{chakrabarty2010approximability} and \cite{devanur2012online} to guide a simple tatonnement style local-search algorithm. 
In fact, one can view our algorithm
as a convex generalization of the LP primal-dual $(4/3)$-approximation algorithm for budget-additive functions in \cite{chakrabarty2010approximability}. This algorithm begins by initializing dual prices greedily and then continuously lowers prices, defecting items to the highest bidder throughout. 
The $4/3$ approximation ratio is then derived from a set of algebraically obtained equations whose solution express the minimum ratio ensuring that as items defect, no agent ends up overspending on an allocation where they received less than what they are given in an optimal fractional solution (which guarantees prices never need to be raised). Clearly  it would be difficult (or impossible, as $v_i(\cdot)$ need not have a closed form)  to derive such a ratio for a set of generic concave-additive functions. Thus for our main technical insight, we show this algebraic approach can be bypassed via more elegant geometric arguments that coincide  directly with the definition of our curvature parameter $\mu_i$.

Additionally, we show our approximation is tight among algorithms that  utilize the natural assignment convex program. 
In particular, we establish the integrality gap of the $\ICA$ convex programming relaxation is precisely our multiplicative curvature parameter.

\begin{theorem}
\label{thm:int-gap}
Consider an instance $\ICA$ where each agent has valuation function $v(\cdot)$ and $u = \max_{i,j}u_{i, j}$. Then the integrality gap of the $\ICA$ assignment convex program is the multiplicative curvature of the instance $\mu$ (determined by $v(\cdot)$ and $u$). 
\end{theorem}

As an application of this result, we note that concave valuations are often examined in the special case {\em piecewise-linear functions} (e.g., see \cite{anari2018nash, chaudhury2018fair, vazirani2011market}), since any continuous concave function can be closely approximated by one that is piecewise-linear. 
In the context of $\ICA$, this corresponds to the following definition: each agent valuation function $v_i(\cdot)$ is defined over on a sequence of conjoined $\lambda_i$ line segments.  
Let $x_{i,k}$ denote the transition point on the $x$-axis between the $k$-th and $(k+1)$-th segments (where $x_{i,0} = 0$). For any such function, we give an algorithm that achieves an approximation ratio of at most $(1+\epsilon)4/3$, as long as the maximum utility gained for a single item is at most the length (along the $x$-axis) of any segment of the piecewise-linear function.

\begin{theorem} 
\label{thm:pl-mult}
Consider an ICA instance where $v_i(u_i)$ is a linear piecewise function such that $\min_{k \in [\lambda_i-1]}\left(x_{i,k+1} - x_{i,k}\right) \geq \max_{j} u_{ij}$.
Then there exists an algorithm whose approximation ratio is $(1+\epsilon)\max_{i}\mu_i \leq (1+\epsilon)4/3$. 
\end{theorem}

For the special case of a budget-additive function where agent $i$ has utility cap $c_i$, the condition required by Theorem~\ref{thm:pl-mult} is equivalent to the standard assumption that $\max_{j} u_{i,j} \leq c_i$. Thus, our bound essentially (i.e., barring the $4/3 - \delta$ approximation for a small constant $\delta$ via the configuration LP in \cite{kalaitzis2015configuration}) matches the state-of-the-art approximation for budget-additive valuations but in the more general setting of piecewise-linear functions. 

We also note that as $\max_{j}u_{ij} \rightarrow 0$, $\mu_i$ tends towards 1. Therefore, for so-called ``small bids" instances, our algorithm outputs near-optimal allocations. To the best of our knowledge, the best algorithm for the small-bids setting is  the online fractional algorithm given in \cite{devanur2012online}, which
for many concave-additive functions achieves a ratio $\gg 1$.\footnote{For example, when $f(x) = x^\alpha$ for $\alpha \in (0,1]$, the algorithm's approximation is $\alpha^{-\alpha}$ (e.g., when $\alpha = 1/2$, $ \alpha^{-\alpha} \approx 1.414$).}

\subsubsection{Additive Guarantees}
\label{subsub:add-guar}

Our multiplicative algorithm and integrality gap example extend to obtaining a tight additive guarantee for $\ICA$ (outlined in Section 
\ref{subsec:additive}, Theorem \ref{thm:add-alg-bound}).
For the main application of our techniques, we show these additive bounds provide new multiplicative approximations for the Nash-welfare objective with asymmetric agents.

As was the case for our multiplicative results,  our additive bound also scales with the integrality gap of the assignment convex program relaxation. 
It has been known that the Nash-welfare objective has an integrality gap of $\Omega(n)$, even for symmetric agents. In fact, the original techniques of Cole and Gkatzelis \cite{cole2015approximating} were directly aimed at circumventing the linear integrality gap. Given that our analysis competes against the optimal fractional objective of the standard assignment CP, the $\Omega(n)$ integrality gap remains a barrier to directly applying our techniques. 


However, another proposed alternative for naturally handling the large gap is to examine agents with {\em smooth valuations} \cite{fain2018fair, fluschnik2019fair}.
In the smooth valuation setting, we give each agent a (potentially fractional) copy of her favorite item at the outset of the allocation, which relaxes the degree to which the objective penalizes under allocations.
Smoothing the valuations essentially ``shifts'' them to the more well-behaved portions of the log objective, but this setting still captures part of the technical challenge of standard (non-smooth) asymmetric objective. 
For example, the algorithm in \cite{garg2020approximating} still has an approximation ratio of $\Omega(n)$ even for smooth valuations \cite{garg2020approximating}.\footnote{One can verify that algorithm in  \cite{garg2020approximating} still has an approximation of $\Omega(n)$ for the tight example instance (Section 6.3) when modified so that each agent gets an additional copy of her most-valued item at the outset.}

More formally, we define an mooth instance with asymmetric agents as follows. First observe that for additive valuations $v_i(u_i) = u_i$, we can scale the objective of each agent $i$ by $(\max_j u_{i,j})^{-\eta_i}$ without changing the approximation factor of the algorithm. 
Therefore, wlog we can assume that $\max_{j}u_{i,j} = 1$. We then define the smooth version of the objective to be $\left(\prod_i (u_i + \omega )^{\eta_i}\right)^{1/\eta}$, where $\omega \in (0,1]$ denotes the {\em smoothing parameter} for the instance. By applying our techniques, we obtain the following result.

\begin{theorem}
\label{thm:snsw}
Consider an instance of Asymmetric  Nash Welfare Maximization with smooth additive valuations. Then there exists an algorithm that runs in time $O(nm^2/(\epsilon \omega))$ that achieves an approximation of $O(e^{\epsilon}/(\omega\ln(1+1/\omega)))$ for smoothing parameter $\omega \in (0,1]$. 
\end{theorem}

Thus for any constant smoothing parameter $\omega$, we obtain a constant approximation for the problem. For example when $\omega = 1$, the approximation ratio of the algorithm is $\approx 1.061$ as $\epsilon \rightarrow 0$. This result also extends to piecewise-linear valuations, yielding an $O(1)$ approximation under the same assumptions as Theorem \ref{thm:pl-mult} when $\omega = \Omega(1)$ (stated formally in Section \ref{subsub:anw-pl}, Theorem \ref{thm:anw-pl}).

Furthermore for additive valuations, the resulting algorithm has an interesting combinatorial interpretation, 
which we call the {\em Weighted Bang-Per-Buck} (WBB) algorithm. 
Many of the aforementioned approximations for the symmetric case (e.g., \cite{anari2018nash, barman2018finding, chaudhury2018fair, garg2018approximating}) also use tatonnement-style algorithms that adjust prices $p_j$ for each item $j$, maintaining throughout that
each item always assigned to a maximum ``bang-per-buck'' (MBB) agent, i.e., agents $i$ such that the ratio $u_{i,j}/p_j$ is maximized. 
In our WBB algorithm, we instead adjust a uniform bid $b_i$ that agent $i$ makes for {\em all} items, but then each item is assigned based on maximum {\em weighted} bang-per-buck ratios, i.e., item $j$ is assigned to the agent that maximizes $(\eta_i u_{i,j})/b_i$. To the best of our knowledge, weighted bang-per-buck ratios have yet to be considered in the context of algorithm design for asymmetric Nash welfare. We hope this concept and interpretation proves useful for making progress in the very challenging non-smooth case. 

\subsection{Related Work}


There is an extensive body of work that has examined the approximability of  welfare maximization for the special case of budget-additive functions. 
A series of results \cite{garg2001approximation, andelman2004auctions, azar2008improved, chakrabarty2010approximability, kalaitzis2015configuration} improved the best approximation to $4/3 - \delta$ for a small constant $\delta$, while the hardness lower bound currently sits at 16/15 \cite{chakrabarty2010approximability}. 
Budget-additive valuations have also been heavily studied in the online setting (typically called the {\em Adwords} problem in this context, motivated by applications in internet advertising), often for divisible items or with the small bids assumption ($u_{i,j} \rightarrow 0)$. 
Please see the excellent survey by Mehta \cite{mehta2013online} for an overview of this area.
We note the algorithm of Devanur and Jain \cite{devanur2012online} for concave-additive valuations with fractional items is viewed as the state-of-the-art approximation in the online setting. 
Their analysis also expresses its approximation as a parameter that depends on the curvature of the concave functions for a  given instance. 
(In the fractional setting, this parameter is instead determined by the solution to a differential equation that tends towards 1 as the function becomes more linear.)

As discussed earlier, there has been an explosion of work on the prolem of Nash Welfare Maximization for indivisible items since the seminal result of Cole and Gkatzelis~\cite{cole2015approximating}, 
who gave the first constant approximation for symmetric agents with additive valuations, achieving a ratio of $2e^{1/e} \approx 2.889$. Later, Cole et al.~\cite{cole2017convex} gave a tight factor 2 analysis of the same algorithm.
 Barman et al.~\cite{barman2018finding} improved the approximation to $e^{1/e} \approx 1.444$, matching the integrality gap example in \cite{cole2017convex}. Since these results, the symmetric case has been examined for a variety of more general valuation functions \cite{garg2018approximating, anari2018nash, garg2021approximating, li2021constant, chaudhury2018fair, barman2021approximating}.
 
 The Nash-welfare objective has also been extensively studied for its appealing fairness properties (e.g.,  Caragiannis et al.~\cite{caragiannis2019unreasonable} call the objective ``unreasonably fair'').  Conitzer et al.~\cite{conitzer2017fair} first gave three relaxations of the proportionality notion of fairness and showed that an optimal solution to the objective satisfies/approximates all three relaxations. For additional results on fairness in relation to Nash welfare, see \cite{barman2018finding, plaut2020almost, mcglaughlin2020improving}. The case of smooth agent valuations was first considered by Fain et al.~\cite{fain2018fair}, also in the context of fairness. Fluschnik et al.~\cite{fluschnik2019fair} showed that this objective is NP-hard to maximize.

Finally, we mention that $\ICA$ is a special case of {\em Submodular Welfare Maximization}, 
where each agent's valuation is given as a general submodular set function.
The optimal approximation in this general setting is $e/(e-1) \approx 1.581$ \cite{vondrak2008optimal, khot2005inapproximability}.
A recent line of work also examined  submodular maximization with {\em bounded curvature}  \cite{conforti1984submodular, Sviridenko2015OptimalAF, vondrak2010submodularity}. 
In this setting, Sviridenko et al.\ showed an optimal approximation of $e/(e-c)$~\cite{Sviridenko2015OptimalAF}, where $c$ is a measure of the {\em total curvature} of the submodular valuation function of the instance.\footnote{Intuitively, $c \in [0,1]$ ($c \rightarrow 1$ implies more curvature) measures the multiplicative gap between the marginal gain of receiving an item after first receiving all other items, versus the marginal gain from only receiving the item.} 
We note that since $c$ is a total curvature parameter (in contrast to our local curvature parameters), many natural concave-additive functions will have a total of curvature of $c=1$.\footnote{For example, for any sufficient large $\ICA$ instance where $v_i'(u_i) \rightarrow 0$ as $u_i \rightarrow \infty$, the total curvature $c$ will be 1.}  Thus in such instances, this bound does not improve upon the $e/(e-1)$ approximation obtained in the general submodular function case.

\subsection{Paper Organization}
The organization of the paper is as follows. In our preliminaries  (Section \ref{sec:prelim}), we outline the convex program relaxation  and dual program of $\ICA $ defined by Devanur and Jain in \cite{devanur2012online} (Section \ref{subsec:prelim-convex-dual}), and then formally define our multiplicative and additive curvature parameters $\mu_i$ and $\alpha_i$ (Section \ref{subsec:lcb}).
For our results, we first present our multiplicative (Sections \ref{sec:pd-alg-def} and \ref{subsec:pd-analysis}) and additive (Section \ref{subsec:additive}) guarantees for general instances of $\ICA$ in Section \ref{sec:ica-algo}.  
Next, we present our tight integrality gap example in Section \ref{sec:int-gap}. 
We conclude in Section \ref{sec:apps} with the two key applications of our techniques, showing our additive bounds for $\ICA$ translate to constant approximations for smooth-asymmetric Nash Welfare Maximization with additive valuations (Section \ref{subsec:smooth-anw}), and outlining the extension to smooth Nash welfare with piecewise-linear valuations (Section \ref{subsub:anw-pl}). 
Finally, we examine our multiplicative guarantees in the special case of piecewise-linear valuations, obtaining an approximation of $(1+\epsilon)4/3$ (Section \ref{subsec:app-pl-welfare}).

\section{Preliminaries}
\label{sec:prelim}

\subsection{Convex Program Formulation and Dual Program}
\label{subsec:prelim-convex-dual}

Recall that every agent $i$ has a non-decreasing concave valuation function $v_i(\cdot)$. Our algorithm utilizes the natural assignment convex program for the problem, which we will refer to as 
$\ICA$-CP: 
\begin{gather*}
    \textnormal{($\ICACP$)}: \max \sum_i v_i(u_i) \\
    \forall i: u_i = \sum_j u_{i,j}x_{ij} \\
    \forall j: \sum_i x_{i,j} \leq 1 \\
    \forall i,j: x_{i,j} \geq 0 
\end{gather*}
The algorithm is primal-dual in nature, and uses the dual program which was defined for the online variant of the problem in \cite{devanur2012online}:
\begin{gather*}
    \textnormal{($\ICAD$)}: \min \sum_i y_i(t_i) + \sum_j p_j \\
    \forall i,j: p_j \geq u_{i,j}v_i'(t_i) \\
    \forall i,j: t_i, p_j \geq 0 
\end{gather*}
where $y_i(t_i) = v_i(t_i) - t_iv_i'(t_i)$ is as defined $y$-intercept of the tangent to $v_i$ at $t_i$. Thus we have the following lemma.

\begin{lemma}[shown in \cite{devanur2012online}] \label{lem:duality}
The above convex programs form a primal-dual pair. That is, any feasible solution to $\ICAD$ has objective at least that of any feasible solution to $\ICACP$.
\end{lemma}

\subsection{Local Curvature Parameters}
\label{subsec:lcb}
As stated in the introduction, both the definition and guarantees provided by our algorithm depend on parameters that measure the local curvatures of each agent valuation function, both in multiplicative (we will use $\mu$) and additive senses ($\alpha$). To define these parameters, let
\begin{equation}
\label{eq:lcb-slope}
\sigma_{i}(z,w) := \frac{v_i(z+w) - v_i(z)}{w},
\end{equation}
be the slope of the lower-bounding  secant line that intersects $v_i$ at points $(z,v_i(z))$ and $(z+w, v_i(z+w))$. Define the  {\em local multiplicative curvature}\footnote{For some functions, a fixed $z^*$ that is an $\arg\max$ in \eqref{eq:mult-lcb} may not exist, and therefore in such cases the $\max$ in the definition should be replaced by a supremum. Such cases can be handled in our analysis by introducing limits when necessary.} of a function $v_i$ at point $z$ with $x$-width $w > 0$ to be:
\begin{equation}
\label{eq:mult-lcb}
\mu_{i}(z, w) := \max_{z^* \in \left(0,w\right)} \left[\frac{v_i(z+z^*)}{v_i(z) + z^*\sigma_i(z,w)} \right].
\end{equation}
Informally, $\mu_{i}(z, w)$ measures the largest multiplicative gap
between a point $z+ z^*$ on the lower bounding secant line and the function evaluated at $z+ z^*$. The definition of $\mu_{i}(z, w)$
is illustrated in Figure \ref{fig:lcb}. The overall local multiplicative curvature for agent $i$ is then defined to be $\mu_i := \max_{z, u_{i,j}} \mu_{i}(z, u_{i,j})$. 

Similarly, we define the {\em local additive curvature} for an agent at point $z$ with $x$-width $w$ to be: 

\begin{equation}
\label{eq:add-lcb}
\alpha_{i}(z, w) := \max_{z^* \in \left(0,w\right)} \left[v_i(z+z^*) - (v_i(z) + z^* \sigma_i(z,w))\right],
\end{equation}
where we again let $\alpha_i := \max_{z, u_{i,j}} \alpha_{i}(z, u_{i,j})$.

\begin{figure}
\centering
\includegraphics[scale=0.7]{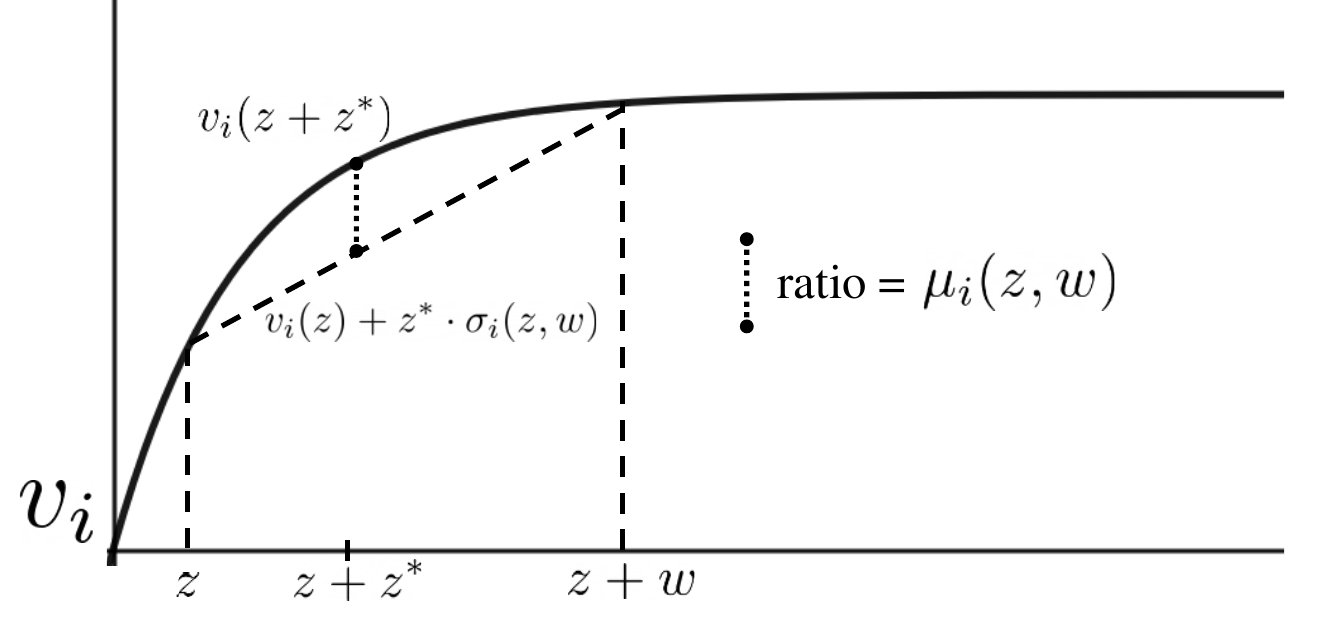}
\caption{{\small Illustration of the definition of the multiplicative local curvature at point $z$ with width $w$ for function $v_i$ (denoted $\mu_{i}(z,w)$).}}
\label{fig:lcb}
\end{figure}

\section{Approximation Algorithms for $\ICA$} \label{sec:ica-algo}
In this section we define our $(1+ \epsilon)\max_i \mu_i$-approximation algorithm for $\ICA$.
For simplicity, throughout the section we assume the algorithm has knowledge of the value of each $\mu_i$ for any given set of valuation functions. 
At the end of the section, we will briefly discuss how the algorithm can be redefined so that algorithm does not need knowledge of $\mu_i$.

\subsection{Algorithm Definition}
\label{sec:pd-alg-def}

As was done for the algorithm in \cite{chakrabarty2010approximability} for budget-additive functions, it will be useful to partition the cost of the dual solution according to algorithm's current assignments. Let $A_i$ denote the current set of items assigned to agent $i$ by the algorithm, and let $t_i$ be the current dual variable maintained by the algorithm for agent $i$. Define
\begin{equation}
\label{eq:dual-agent}
D(u_i) = y_i(t_i) + u_i v_i'(t_i)
\end{equation}
to be the utility for agent $i$ but instead evaluated according 
to the tangent line in the dual objective taken at point $t_i$. The algorithm will maintain that at any point, each item $j$ will be assigned to the bidder $i$ that maximizes 
$u_{i,j}v_i'(t_i)$ (and will reassign an item if this doesn't hold). We call such an assignment a {\em proper assignment.} 


\begin{definition} \label{def:proper-allocation}
Given a dual solution $v$, an item $j$ is said to be properly assigned if $j$ is assigned to agent $\arg\max_{i}(u_{i,j}v_i'(t_i))$. Otherwise, item $j$ is said to be improperly assigned. 
\end{definition}

In an allocation where all item's are properly
assigned, we can obtain the following characterization of the dual objective.

\begin{lemma}
\label{lem:alg-dual-prop}
Fix a point in the algorithm with primal and dual variables $u_i$ and $t_i$ for each agent $i$. If all items are properly allocated, then (i) setting $p_j = \max_{i} v'_i(t_i)u_{i,j} $ forms a feasible solution (with the $t_i$) to the dual $\ICA$-D, and (ii)
the objective of the dual can be expressed as $\sum_{i} D(u_i)$. 
\end{lemma}

\begin{proof}
Note that (i) follows immediately from the dual constraints. 
To see (ii), observe that
\begin{equation}
\label{eq:alg-dual-prop2}
\sum_j p_j = \sum_i \sum_{j \in A_i} (u_{i,j} v_i'(t_i)) = \sum_i ( u_i v_i'(t_i)).
\end{equation}
Adding $\sum_i y_i(t_i)$ to both sides of  \eqref{eq:alg-dual-prop2}, on the LHS we obtain the dual objective,
and on the RHS we obtain $\sum_{i} D(u_i)$, as desired.
\end{proof}

We can now define our algorithm (given formally in  Algorithm \ref{alg:pd}). At any point the algorithm maintains a 
setting of dual variables $t_i$ for all agents $i$ and variables $p_j$ for all items $j$. 
Each $t_i$ variable is initialized to be 0, and then items are properly assigned accordingly.
The algorithm proceeds by continuously increasing $t_i$ values (thus decreasing $v_i'(t_i)$), allowing items to defect if they are no longer properly assigned. 
The goal of the algorithm is to eventually is obtain an allocation such that $D(u_i)/ v_i(u_i)  \leq \mu_i$ for all agents. 
Under this condition, as long items remain (close to) properly assigned upon the algorithm's termination, Lemmas~\ref{lem:duality} and~\ref{lem:alg-dual-prop} imply that the approximation ratio of the algorithm is at most $\max_i \mu_i$.


\begin{algorithm}
\caption{Multiplicative $(\max_i \mu_i$)-Approximation Algorithm for $\ICA$}
\label{alg:pd}
Initialize all $t_i = 0$ \\
Allocate each item $j$ to agent $\arg\max_{i} \left(u_{i,j}v_i'(0)\right)$ \\
\While{there exists a agent $i$ such that $D(u_i)/ v_i(u_i) > \mu_i$}{
    \While{$D(u_i)/ v_i(u_i) > \mu_i$}{
    \eIf{there is an item $j$ such that $j$ is assigned to $i$ and $i \neq \arg\max_{k} \left( u_{k,j}  v_k'(u_k)\right)$}{
    Reassign $j$ to agent $\arg\max_{k} \left( u_{k,j}  v_k'(u_k)\right)$
    }
    {Increase tangent point $t_i$ until $v_i'(t_i)$ decreases by a factor of $1/(1+\epsilon)$}
    }
}
Output resulting allocation $u_i$ for all agents
\end{algorithm}

\subsection{Analysis}
\label{subsec:pd-analysis}
Our algorithm seeks to find an allocation such that, for every agent $i$, the inequality $v_i(u_i) / D(u_i) \geq \mu_i$ is satisfied. To this end, we establish the following terminology.

\begin{definition}
There are two types of agents such that $D(u_i)/v_i(u_i) > \mu_i$: either $u_i < t_i$ or $u_i < t_i$. 
Call such agents under allocated and over allocated, respectively. 
\end{definition}

The main technical hurdles for our more general setting of concave-additive functions is showing that no agent becomes under allocated and establishing the claimed run time bound. We first state and prove a useful a property of monotone-concave functions. 

\begin{lemma}
\label{fact:concave-fact}
Let $f:\mathbf{R}_+\rightarrow \mathbf{R}_+$ be a monotone concave function. Suppose $\ell_1(\cdot)$ and $\ell_2(\cdot)$ are the equations of two lines tangent to $f$ at points $(t_1, f(t_1))$ and $(t_2, f(t_2))$, respectively. If there exists $x \geq \max(t_1, t_2)$ such that $\ell_2(x) \leq \ell_1(x)$, then (i) $t_1 \leq t_2$ and (ii) $\ell_1(\tilde x) \geq \ell_2(\tilde x)$ for all $\tilde x \leq t_1$.
\end{lemma}

\begin{proof}
We first show (i). For sake of contradiction suppose $t_1 > t_2$. Since $f$ is monotone and concave, it follows that:

\begin{equation}
\label{eq:concav-fact-1}
f(t_1) < \ell_2(t_1) = f(t_2) + (t_1 - t_2)f'(t_2).
\end{equation}
Therefore, we have the following:
\begin{alignat*}{2}
\ell_1(x)  = f(t_1) + (x- t_1)f'(t_1)  &< f(t_2) + (t_1 - t_2)f'(t_2) + (x - t_1)f'(t_2) \\
& = f(t_2) + (x - t_2)f'(t_2) \\
& = \ell_2(x),
\end{alignat*}
where the inequality follows from \eqref{eq:concav-fact-1} and the fact that $f'(t_1) < f'(t_2)$. This directly contradicts the assumption in the lemma, thus showing (i). Given $t_1 \leq t_2$, (ii) then follows by a similar argument. 
\end{proof} 
\begin{lemma} \label{lem:under-bad}
Throughout the algorithm, an agent never becomes under allocated. In particular, if $u_i < t_i$ then  $D(u_i)/ v_i(u_i) \leq \mu_i$.
\end{lemma}

\begin{proof}
At the start of the algorithm $t_i = 0$ for all agents, and so no agent can be under allocated at the outset of the algorithm.
Therefore, the only point at which 
an agent $i$ with total utility $u_i$ could potentially become under allocated is when some item $j$ is reassigned to another agent on Line 6 of the algorithm such that after the reassignment $u_i - u_{i,j} < t_i$. Fix 
such a point in the algorithm. 



Let $s(x)$ denote the equation for the secant line that passes through $v_i$ at points $(u_i- u_{i,j}, v_i(u_i- u_{i,j}))$ and $(u_i, v_i(u_i))$. 
More formally, let $\sigma_i' := \sigma_i(u_i - u_{i,j}, u_{i,j})$, 
where the definition of $\sigma_i(\cdot)$ is given by Equation \eqref{eq:lcb-slope}.
Then the equation for $s(x)$ is

\begin{equation}
s(x) = \sigma_i' x + v_i(u_i - u_{i,j}) - \sigma_i'  (u_i - u_{i,j}).
\end{equation}

Let $\tilde\mu_i := \tilde\mu_i(u_i-u_{i,j})$ denote the local multiplicative curvature of $v_i(\cdot)$ at $u_i-u_{i,j}$  , and let $z^*$ be value (given in Equation \eqref{eq:mult-lcb}) that determines $\tilde\mu_i$.
Based on the definition of $\tilde\mu_i$, observe that if we scale $s(x)$  by a factor of $\tilde\mu_i$, we obtain an equation for a line that is tangent to 
$v_i$ at the point ($u_i - u_{i,j} + z^*, v_i(u_i - u_{i,j}+ z^*))$. 
Denote the equation for this line as $\tilde{s}(x) = \tilde\mu_is(x)$. 
Since by definition $s(u_i) = v(u_i)$, it follows that

\begin{equation}
\label{eq:secant-bound}
\tilde s(u_i) = \tilde\mu_i s(u_i) = \tilde\mu_iv_i(u_i)  \leq \mu_i v_i(u_i) < D(u_i),
\end{equation} 
where the first inequality holds by definition of $\mu_i$ and the last inequality holds because $i$ entered the loop on Line 3.  Given Inequality \eqref{eq:secant-bound}, we can  apply  Lemma \ref{fact:concave-fact} to show that $D(u_i - u_{i,j}) \leq  \tilde s(u_i - u_{i,j})$. In particular, by setting $\ell_1(\cdot) = s(\cdot)$, $\ell_2 = D(\cdot)$ and $x = u_i$, the lemma implies both that $t_i \leq z^*$ and $D(u_i - u_{i,j}) \leq   \tilde s(u_i - u_{i,j})$. Thus it follows:
\begin{equation}
\label{eq:ratio-est}
\frac{D(u_i-u_{i,j})}{v_i(u_i-u_{i,j})} =\frac{D(u_i-u_{i,j})}{s(u_i-u_{i,j})} \leq \frac{\tilde s(u_i-u_{i,j})}{s(u_i-u_{i,j})} = \tilde\mu_i \leq \mu_i,
\end{equation} 
 which establishes the lemma.  Note the first equality follows from the definition of $s(x)$, and the first inequality follows from $D_{t_i}(u_i - u_{i,j}) \leq   \tilde s(u_i - u_{i,j})$.
\end{proof}



Next,  we establish the run time of the algorithm. For simplicity, we will assume that $\epsilon$ is selected such that for all agents $\mu_i \geq 1+\epsilon$. This is possible as long as $\mu > 1$. (If $\mu_i = 1$, then $v_i(u_i) = D(u_i)$, and thus the agent is never under or over allocated.)

\begin{lemma}
\label{lem:alg-pd-term}
Let $\umaxi = \sum_{j} u_{i,j}$ denote the maximum possible spend for fixed agent $i$, and let 
\begin{equation*}
\rmax = \max_i \left(\frac{v_i'(0)\umaxi}{v_i(\umaxi)}\right)
\end{equation*} 
denote the maximum ratio (over all agents $i$) between the total additive utility evaluated along the tangent line at $v_i'(0)$, versus the total spend evaluated with $v_i(\cdot)$. If for all agents $\mu_i \geq 1+\epsilon$, then the algorithm terminates in time $O(mnT\ln(\rmax/\epsilon)/\epsilon)$, 
where $T$ is the time needed to perform the update of $v_i'(t_i)$ on Line 8.
\end{lemma} 

\begin{proof}

We begin by showing that once the update on Line 8 occurs  $O(\ln(\rmax/\epsilon)/\epsilon)$ times for a fixed agent $i$, 
then agent $i$ cannot be under or over allocated for the remainder of algorithm's execution 
(i.e., the algorithm does not again enter the while loop on Line 3 for agent $i$). Define $H_i$ as the following:

\begin{equation*}
H_i = \frac{1}{\ln(1+\epsilon)}\cdot \ln\left(\frac{v'_i(0)\umaxi}{\epsilon v_i(\umaxi)}\right) \leq \frac{\ln(\rmax/\epsilon)}{\ln(1+\epsilon)} = O\left(\frac{\ln(\rmax/\epsilon)}{\epsilon}\right).
\end{equation*} 
After $H_i$ Line 8 updates to agent $i$, 
we can bound $v'_i(t_i)$ as follows:

\begin{equation} 
\label{eq:runtime-1}
 v_i'(t_i) = \frac{v_i'(0)}{(1+\epsilon)^{H_i}}
= \frac{\epsilon  v_i(\umaxi)}{\umaxi} 
\leq \frac{(\mu_i-1)v_i(\umaxi)}{\umaxi}
\end{equation}
where the inequality follows from our assumption 
$\mu_i \geq 1+\epsilon$. Rearranging \eqref{eq:runtime-1}, we obtain
\begin{equation}
\label{eq:runtime-2}
v_i(\umaxi) + \umaxi v_i'(t_i) \leq \mu_i v_i(\umaxi).
\end{equation}
 Recall that by Lemma~\ref{lem:under-bad}, an agent can be never be under allocated. Therefore we may assume that $t_i < u_i$. We can then upper bound $D(u_i)$ as follows:
\begin{align*} 
D(u_i) &= v_i(t_i) + (u_i - t_i)\cdot v_i'(t_i) \\
& \leq v_i(\umaxi) + \umaxi v_i'(t_i) \leq \mu_iv_i(\umaxi),
\end{align*}
where the first inequality follows from the fact 
$t_i \leq u_i \leq \umaxi$, and the last inequality follows from \eqref{eq:runtime-2}. This means agent $i$ cannot be over allocated for the rest of the algorithm, which means agent $i$ will remain paid for.

To complete the argument, notice that once an item $j$ is reassigned on Line 6, the algorithm cannot reassign $j$ again until it increases $t_i$ for some agent $i$ on Line 8. 
Thus, the algorithm can perform at most $m$ reassignments on Line 6 before an update on Line 8 must occur for some agent.
Summing over all agents establishes the lemma.
\end{proof}

\begin{proof}[Proof of Theorem~\ref{thm:pd-alg}]

Notice that after termination, by increasing each dual variable $p_j = \max_{i} v'_i(t_i)u_{i,j}$ by a factor of $(1+\epsilon)$ we obtain a feasible dual solution, since the only items that are improperly assigned are ones allocated to an agent that exited the while loop on Line 3 before the item should have been reallocated on Line 6 in the inner while loop. Let $\OPT$ denote the optimal solution for the instance. By the above argument and Lemma~\ref{lem:duality} we have
\[
\OPT \leq (1+\epsilon)\sum_i D(u_i).
\]
Lemma \ref{lem:alg-pd-term} ensures the algorithm will eventually terminate, and since the algorithm terminates, no agent can be over allocated.  Furthermore, by Lemma \ref{lem:under-bad}, no agent can be under allocated. Therefore for all agents $i$ we have $D(u_i) \leq \mu_i v_i(u_i)$ for every $i$. Combining these inequalities yields
\[
\OPT \leq (1+\epsilon)\sum_i \mu_i v_i(u_i) \leq (1+\epsilon) (\max_i \mu_i)\cdot \sum_i v_i(u_i),
\]
which proves the theorem.
\end{proof}
We conclude the section by briefly discussing how
Algorithm \ref{alg:pd} can still execute without knowledge of $\mu$ and can be adapted to the additive setting.

\subsubsection{No Knowledge of $\mu_i$} The algorithm can be adapted to operate without knowledge of each $\mu_i$ by  (i) repeatedly guess values of $\mu = \max_i\mu_i$, then (ii) check Lines 3 and 4 against $\mu$ instead.
In particular, we can start with an overestimate of $\mu$ 
(i.e., start with $\mu = 1+\epsilon$ and repeatedly raise the guess in increments of $\epsilon$), and then check after all reassignments whether 
or not an agent is under allocated. If no agent ever becomes under allocated, then it follows the algorithm achieves an approximation at its current guess; furthermore, by the Lemma \ref{lem:under-bad}, the algorithm is guaranteed to have no agents become under allocated once the guess reaches the true value of $\mu$. This modification comes at an $O(1/\epsilon)$ factor in the run time of the algorithm and an additional additive $O(\epsilon)$ error in the approximation factor. 

\subsubsection{Adaptation to Additive Guarantee}
\label{subsec:additive} 

As discussed in the introduction, an appealing feature of our geometric-based arguments is that they easily extend to obtain additive guarantees as well. This result is stated formally as follows.

\begin{theorem}
\label{thm:add-alg-bound}
There exists an  algorithm for $ICA$ that achieves an additive bound of $\sum_i \alpha_i + \epsilon$ and runs in time $O(m^2nTv_i'(0)/\epsilon)$, 
where $T$ is the time needed to perform the update of $v_i'(t_i)$ on Line 8 of Algorithm~\ref{alg:pd}. Furthermore, the integrality gap of the corresponding assignment convex program is $\sum_i \alpha = n\alpha$ for instances whose valuations have additive curvature $\alpha$.
\end{theorem}

Unfortunately in a general instance $\ICA$ in the additive setting, our guarantee becomes $\sum_i \alpha_i$ (instead of $\max_i \alpha_i$) since our analysis bounds the objective on a per-agent basis.
We note (i) it can be verified that the multiplicative integrality gap example in the next section (Section \ref{sec:int-gap}) can be adapted to show that this bound is tight, and (ii) as we will see in Section \ref{subsec:smooth-anw}, the $\sum_i \alpha_i$ additive bound will translate nicely to a constant approximation for the smooth asymmetric Nash-welfare product objective, since the additive curvature of the weighted-log objective for each agent $i$ will be $O(\eta_i$) with  smoothing parameter $\omega = \Omega(1)$. (The contribution of $\sum_i \eta_i$ from the additive bound is canceled by the $\frac{1}{\sum_i \eta_i}$ term taken in the overall exponent of the product objective). 

The adaption of the algorithm and analysis from Sections \ref{sec:pd-alg-def} and \ref{subsec:pd-analysis} to the additive setting (to prove Theorem \ref{thm:add-alg-bound}) uses many of the same definitions and arguments, so we will only highlight the key differences here:

\begin{itemize}

\item In the algorithm, the while-loop conditions on Lines 3 and 4 are changed to instead be $D(u_i) -  v_i(u_i) > \alpha_i$. The update on Line 8 is changed so that $v_i'(t_i)$ instead decreases by (an absolute amount) $\epsilon/m$. 

\item In our analysis, the arguments in Lemmas \ref{lem:under-bad} work in the additive setting
since the relevant quantities are completely relative to each other, and therefore ratios can be replaced by differences while preserving the logic. The feasibility (adjusting $p_j$ by $\epsilon/m$) and run-time arguments can also be easily adapted to obtain the bounds stated in the theorem.\footnote{In the case of general $\ICA$ instance, the algorithm runs in pseudo-polynomial time based on the value of $v'(0)$. We note that in the our Nash-welfare application, this term will be constant, and thus runs in strongly-polynomial time for a fixed $\epsilon$ and smoothing parameter $\omega$.}
\end{itemize}

\section{Integrality Gap of $\ICA$ Convex Program}
\label{sec:int-gap}

We now prove Theorem \ref{thm:int-gap}. In particular, we show that for any fixed monotone concave value function $v(\cdot)$ and maximum bid value $u$ with local multiplicative curvature $\mu$, we can construct an instance of $\ICA$
such that the optimal fractional solution to $\ICACP$ has objective $\mu$ times that of any integral assignment. 


\paragraph{Instance Construction.} Let $z$ be the $\arg\min$
that defines $\mu$ for function  $v(\cdot)$, and let $z^*$ be the $\arg\max$ that defines $\mu_i(z, u)$ in the definition of $\mu_i$, as shown on the left side of Fig.~\ref{fig:int-gap}.
Without loss of generality, we can assume $z^*/u$ can be expressed as 
rational number $\beta/\gamma$ where $\beta, \gamma \in \mathbb{Z}^+$, since any irrational number has arbitrarily close rational  approximation. (In which case our instance construction can be taken such a limit to obtain the desired bound).

We construct our instance as follows. 
The valuation function of every agent is $v(\cdot)$. 
There are $\gamma$ agents and $\beta + \gamma \cdot \lceil \frac{z}{u} \rceil$ items in total. Of the items, $\beta$ of them are ``public,'' i.e., for all agents $i$ we have have $u_{i,j} = u$. Call this subset of items $M_{\pub}$.  The remaining items $\gamma \lceil\frac{z}{u}\rceil$ items are partitioned among the $\gamma$ agents so that each agent $i$ receives a set of $\lceil\frac{z}{u}\rceil$ ``private'' items $M_i$. For each item $j$ in the first $\lfloor \frac{z}{u}\rfloor$ of these private items in $M_i$, we set $u_{i,j} = u$. 
For the remaining item (if there is one) $j'$ in $M_i$, we set $v_{i,j'} = z - \lfloor \frac{z}{u}\rfloor u < u$.
For all other agents $i' \neq i$, $u_{i',j} = 0$ for all $j \in M_i$. 
Based on this construction, the sum of $u_{i,j}$ over all $j \in M_i$ is equal to $z$, and the maximum bid value in the instance is indeed $u$. 
This completes the construction, which is illustrated in Figure \ref{fig:int-gap}. 

\begin{figure}
\centering
\includegraphics[height=5cm]{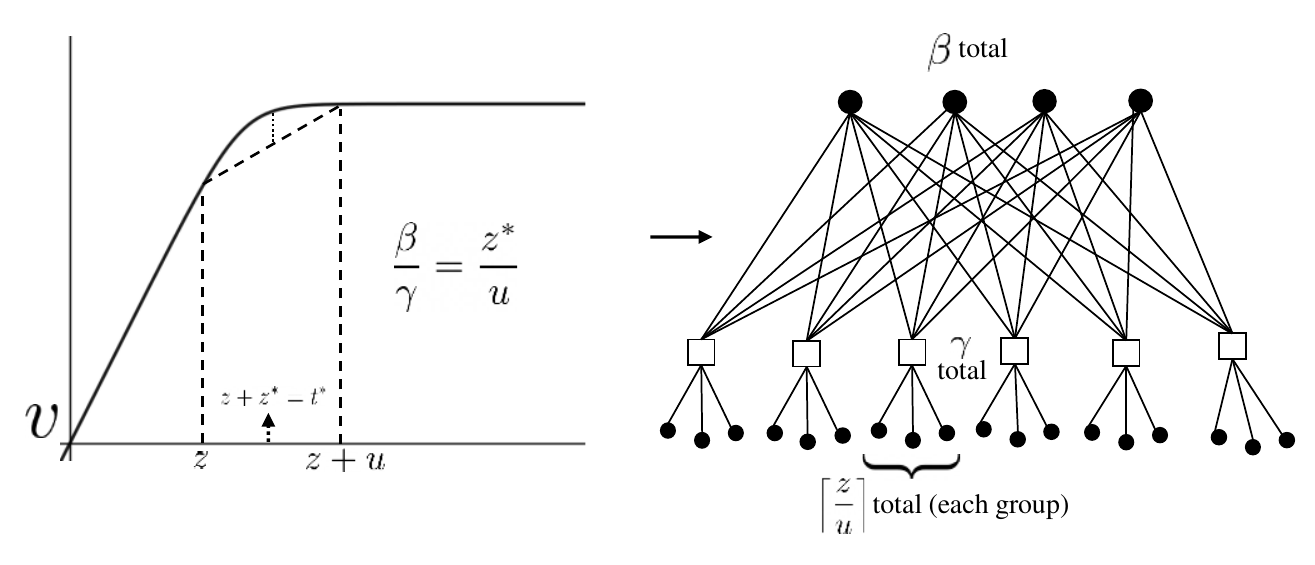}
\caption{{\small Illustration of integrality gap construction. Squares correspond to agents and circles correspond to items.}}
\label{fig:int-gap}
\end{figure}

\paragraph{Analysis.} Let $\OPT_F$ and $\OPT_I$ denote 
the optimal fractional and integral solutions, respectively, to the $\ICACP$ convex program for the above instance construction. 
We first show that $\OPT_F$ is obtained by evenly splitting the spend of the $\gamma$ agents among the $\beta$ public items. 
To do this, we construct a feasible dual solution $\DUAL$ to the $\ICAD$ program that has objective equal to that of $\OPT_F$.
By Lemma \ref{lem:duality}, it  then follows $\OPT_F$ is an optimal fractional solution. 
Constructing this dual solution will also be useful for showing the desired integrality gap, as it will be easier to relate the objective of $\OPT_I$ to $\DUAL$. Note that to simplify notation, for the remainder of the section we denote $t^* := z + z^*$. 

Observe that in the solution $\OPT_F$ specified above, each agent spends a total of $u\beta/\gamma$ on public items.
By the definitions of $\beta$ and $\gamma$, we have $u\beta/\gamma = z^*$.  Each agent also receives a total spend of $z$ from from their private items.  
Therefore, since there are $\gamma$ agents in the total, the objective of $\OPT_F$ is the following:

\begin{equation}
\label{eq:int-gap-opt}
\OPT_F = \gamma v(z+z^*) = \gamma v(t^*).
\end{equation}
To construct $\DUAL$,  we set $t_i = t^*$ for all agents $i$. 
Since each agent has an identical valuation function $v(\cdot)$, setting variables $p_j = u v'(t^*)$ if $j$ is a public item and setting $p_j = u_{i,j}v'(t^*)$ if $j$ is a private item in $M_i$ produces a feasible solution to $\DUAL$.

To show that the objectives of $\OPT_F$ and $\DUAL$ are the same, first observe that the definitions of $y(\cdot)$ and $v'(\cdot)$, we have the following identities: 
\begin{gather} 
\label{eq:int-gap-dual-1}
y(t^*) + (z+u)v'(t^*) = v(t^*) + (u-z^*)v'(t^*) \\
\label{eq:int-gap-dual-2}
y(t^*) + zv'(t^*) = v(t^*) - z^*v'(t^*).
\end{gather}
 Equations \eqref{eq:int-gap-dual-1} and \eqref{eq:int-gap-dual-2} give two equivalent ways of expressing the value  
of the line tangent to $v(\cdot)$ at $t^*$ evaluated at $x$-coordinates $z+u$ and $z$, respectively. Geometrically, on the LHS, we start at the $y$-intercept of the line, 
and then follow the slope of the tangent for a total $x$-width of $z+u$ (resp.\ $z$). On the RHS, we instead start at $v(t^*) = v(z+z^*)$ (i.e., the tangent point)
and follow the tangent line to $z+u$ (resp.\ backwards to $z$).

From the construction of the instance and the definition of $\DUAL$, the objective of $\DUAL$ can be characterized as follows:
\begin{align*}
\sum_i \left(y(t^*) + \sum_{j \in M_i} p_j\right) + \sum_{j \in M_{\pub}}p_j & = \gamma  (y(t^*) + z  v'(t^*)) +  \beta u v'(t^*) \\ 
& = \beta \cdot (y(t^*) + (z +u) v'(t^*)) +  (\gamma - \beta) \cdot (y(t^*) + zv'(t^*)).
\end{align*}
By applying Equations \eqref{eq:int-gap-dual-1} and \eqref{eq:int-gap-dual-2} to the RHS above, it follows the RHS is equal to: 
\begin{equation}
\label{eq:dual-gap-obj}
\beta \cdot (v(t^*) + (u-z^*)v'(t^*)) + (\gamma-\beta) \cdot (v(t^*)-z^* v'(t^*)), 
\end{equation}
which can simplified as follows:
\begin{align*} 
\gamma v(t^*) + v'(t^*) \cdot (\beta(u-z^*) -(\gamma-\beta)z^*)  & = \gamma v(t^*) + v'(t^*) \cdot (\beta u - \gamma z^*)  \\
& = \gamma v(t^*) \\
& = \OPT_F
\end{align*}
where the second equality follows since $\beta/\gamma = z^*/u$, and the last equality follows from \eqref{eq:int-gap-opt}.
Thus $\OPT_F$ and $\DUAL$ have equivalent objective values.

Now consider $\OPT_I$, which is obtained by assigning a unique public item to $\beta$ of the $\gamma$ agents.
(By a simple exchange argument, the objective cannot increase by assigning multiple public items to the same agent, since $v(\cdot)$ is monotone and concave.)
Each agent that receives a public item spends a total of $z + u$. 
The remaining $\gamma - \beta$ agents spend only a  total of $z$ from solely their private items. 
Thus the objective of the optimal integral solution is:

\begin{equation}
\label{eq:gap-int-obj}
\OPT_I = \beta v(z+u) + (\gamma-\beta) v(z).
\end{equation}

By definition of $\mu$, the gap between $v(z+b)$ and $v(z)$, and the (respective) points 
characterized by Equations \eqref{eq:int-gap-dual-1} and \eqref{eq:int-gap-dual-2} are equal to $\mu$. More specifically, we have
\begin{equation}
\label{eq:gap-int-char}
 \mu = \frac{v(t^*) + (u-z^*)v'(t^*)}{v(z+u)} = \frac{v(t^*) - z^*v'(t^*)}{v(z)}.
\end{equation}

Since \eqref{eq:dual-gap-obj} expresses the objective value of $\DUAL$ (and therefore $\OPT_F$, as well), 
Equations \eqref{eq:dual-gap-obj}, \eqref{eq:gap-int-obj}, and \eqref{eq:gap-int-char} together imply $\OPT_F = \mu\OPT_I$, i.e., the integrality gap of this instance is $\mu$.

\section{Applications}
\label{sec:apps}


\subsection{Constant Approximation for Asymmetric Nash Welfare with Smooth Valuations}
\label{subsec:smooth-anw}
We now apply our techniques to the problem of Nash Welfare Maximization for asymmetric agents with smooth additive additive valuations.
In this problem, each agent $i$ has a weight $\eta_i > 0$, and the goal is to find an allocation that maximizes $ \left(\prod_i (u_i + \omega)^{\eta_i}\right)^{1/\eta}$
where $\eta = \sum_i \eta_i$ is the sum of the agent weights and $\omega \in (0,1]$ denotes the smoothing parameter of the instance.

As discussed in the introduction, observe that we can scale the objective of each 
agent $i$ by $(\max_j u_{i,j})^{-\eta_i}$ without changing the approximation ratio of the algorithm. 
Therefore, wlog for the rest of the section we will assume that $\max_{j}u_{i,j} = 1$ for every agent $i$.
After this scaling, we can think of the smoothing parameter as giving each agent $i$ an initial utility of $\omega \max_{j}u_{i,j}$
at the beginning of the instance. Also for simplicity, throughout the section we assume weights are normalized by dividing them by $\eta$, so $\eta = 1$ (i.e., we bring the $1/\eta$ exponent into each term in the product objective).

\subsubsection{Algorithm Definition}

Our algorithm has a natural combinatorial interpretation  which we call the {\em Weighted Bang-Per-Buck} (WBB) algorithm. To define the algorithm, we first explicitly define the additive curvature parameter $\alpha_i$ in the case where $v_i(u_i) = \eta_i \ln(u_i + \omega)$. Let $\sigma_i(z)$ denote the slope of the lower-bounding secant line that intersects the points $(z, \eta_i\ln(z + \omega))$ and $(z+1, \eta_i\ln(z+\omega + 1))$, given as:
\begin{equation}
\label{eq:sn-slope}
\sigma_{i}(z,1) := \eta_i\ln(z+ \omega + 1) - \eta_i\ln(z+\omega) =\eta_i\ln\left(1 + (z+\omega)^{-1}\right).
\end{equation} 
We then define the local additive curvature bound $\alpha_i$ at $z$ for agent $i$:
\begin{equation}
\label{eq:sn-curv}
\alpha_i(z) := \max_{z^* \in (0, 1)} \left[\eta_i\ln(z + z^* + \omega) - (\eta_i\ln(z + \omega) + z^*\sigma_{i}(z,1))\right]
\end{equation}
The  WBB Algorithm is given below in Algorithm 2.  
Throughout its execution, we adjust a uniform bid $b_i$ each agent $i$ makes for on every item in the instance. The algorithm starts with bids that are underestimates of the optimal dual bids, and thus proceeds by increasing the uniform bid of each agent one at a time, ensuring throughout every item is assigned to a maximum weighted bang-per-buck ratio agent, i.e., an agent that maximizes $(\eta_i u_{i,j})/b_i$. 
The algorithm stops increasing the bid of an agent according an exponential potential function proportional to agent's average {\em unweighted} MBB ratio (which we derive from the while-loop condition from the additive version of the $\ICA$ algorithm given in Section \ref{subsec:additive}).

\begin{algorithm}

\caption{Maximum Weighted Bang-per-buck Algorithm ({\sc WBB})}
Initialize fixed bid $b_i \leftarrow \omega$ for each agent $i$  \\
 \label{alg:WMBB}
Allocate each item $j$ to maximum WBB agent $\arg\max_{i} \left( \frac{ \eta_i u_{i,j}}{b_i} \right)$ \Comment{weighted greedy assignment} \\
\While{there exists an agent $i$ such that $\frac{u_i + \omega}{b_i}< \exp\left(\frac{u_i + \omega}{b_i} - 1 - \alpha_i\right)$}{
    \While{$\frac{u_i + \omega}{b_i}< \exp\left(\frac{u_i + \omega}{b_i} - 1 - \alpha_i\right)$}{
    \eIf{there is an item $j$ assigned to agent $i$ such that $i$ is not $j$'s maximum WBB agent}
    {
     Reassign $j$ to agent $\arg\max_{k}\left(\frac{\eta_{k} u_{k,j}}{b_{k}}\right)$
    }
    {
    Increase agent $i$'s bid to be $b_i \leftarrow \frac{\eta_imb_i}{\eta_i m - \epsilon b_i}$ \\
    }
}  
}
Output resulting allocation $u_i$ for all agents
\end{algorithm}

\subsubsection{Analysis}
To analyze the algorithm, we first argue that the WBB algorithm is equivalent to executing the ICA algorithm for an additive guarantee (as outlined in Section \ref{subsec:additive}). We then derive a closed-form for the local additive curvature $\alpha_i$ in terms of the smoothing parameter $\omega$.

\begin{lemma} 
The WBB algorithm (Algorithm 2) is equivalent to executing the ICA algorithm for an additive guarantee, where in the ICA instance $v_i(u_i) = \eta_i \ln(u_i + \omega)$.
\end{lemma} 

\begin{proof} 
By Lemma \ref{lem:duality}, the  primal and dual program for an ICA instance with $v_i(u_i) = \eta_i \ln(u_i + \omega)$ is given by the following (denoted {\sc ASN-CP} for ``Asymmetric Smooth Nash''):  

\begin{center}
\begin{tabular}{c  c  c | c  c} 
\hspace{5mm} & 
$
\begin{gathered}
    \textnormal{({\sc ASN-CP}):}  \max \sum_i \eta_i\ln(u_i+\omega) \\
    \forall i: u_i = \sum_j u_{i,j}x_{i,j} \\
    \forall j: \sum_i x_{i,j} \leq 1 \\
    \forall i,j: x_{i,j} \geq 0 \\
\end{gathered}
$
& \hspace{1mm} & \hspace{1mm} & 
\vspace{-4mm}
$
\begin{gathered}
    \textnormal{({\sc ASN-D}):} \min \sum_i \eta_i \left(\ln(t_i+ \omega) - \frac{t_i}{t_i+\omega}\right) + \sum_j \beta_j \\
    \forall i,j: \beta_j \geq \frac{\eta_i u_{i,j}}{t_i + \omega} \\
    \forall i,j: t_i, \beta_j \geq 0 \\
\end{gathered}
$
\end{tabular} \\
\end{center} 
\vspace{5mm}



Note that for this application, we denote the dual variable $p_j$ as $\beta_j$, since it is interpreted as the weighted MBB ratio, not the price. In particular, in the WBB algorithm, we substitute the $t_i + \omega$ terms in the {\sc ASN-D} program to be the uniform bid $b_i$ made by agent $i$ for all items. Thus the function $D(u_i)$ becomes: 

\begin{equation} 
\label{eq:Dlogdef}
D(u_i) = \eta_i\left[\frac{u_i + \omega}{b_i} + \ln(b_i)  -1 \right].
\end{equation}

Rearranging \eqref{eq:Dlogdef} the while-loop condition in Algorithm \ref{alg:pd} (which is $v_i(u_i) - D(u_i) > \alpha_i$ for the general additive $\ICA$ algorithm) and  exponentiating, we obtain the while-loop condition in Algorithm~\ref{alg:WMBB} after canceling $\eta_i$ terms. Furthermore, since WBB algorithm maintains an assignment where each item is assigned to the agent with maximum weighted MBB ratio, the variables $t_i = b_i - \omega$ and $\beta_j = \arg\max_{i} \left( \frac{ \eta_i u_{i,j}}{b_i} \right)$ form a feasible dual solution in a manner identical to the ICA algorithm. 
Finally, it is easily seen that the update to bid $b_i$ decreases $v'_i(t_i)$ by $\epsilon/m$:

\begin{equation*}
v_i'^{(2)}(t_i) = \frac{\eta_i}{\frac{\eta_imb_i}{\eta_i m - \epsilon b_i}} = \frac{\eta_i}{b_i} - \frac{\epsilon}{m} = v_i'^{(1)}(t_i) - \frac{\epsilon}{m},
\end{equation*}
where $v_i'^{(2)}(t_i)$ and $v_i'^{(1)}(t_i)$ denote the $v_i'(t_i)$ before and after the update to $b_i$ (respectively). 
\end{proof}

\begin{lemma}
\label{lem:nsw-diff-gap}
The local additive curvature $\alpha_i$ for agent $i$ is given by:

\begin{equation*} 
\alpha_i = \eta_i  \left[ \ln\left(\frac{1}{\omega \ln(1+1/\omega)}\right) + \omega \ln(1+1/\omega) - 1 \right]= O\left(\ln\left(\frac{\eta_i}{\omega\ln(1+\omega)}\right)\right),
\end{equation*}
when valuation function of agent $i$ is $v_i(u_i) = \eta_i \ln(u_i + \omega)$.
\end{lemma}

\begin{proof} 

For a fixed value $z$, let $\gamma = z + \omega$ and let $\Delta(z^*)$ be a function of $z^*$ defined by the expression inside the $\max$ in Equation \eqref{eq:sn-curv} for $\alpha(z)$, given as:

\begin{equation*}
\Delta(z^*) = \eta_i\left(\ln(z^* + \gamma) - \ln \gamma - z^*\ln(1+1/\gamma)\right).
\end{equation*}
Observe $\Delta(z^*)$ is a concave function in $z^*$, since the term $\ln(z^* + \gamma)$ is a concave function of $z^*$
and the remainder of the expression is a linear function in $z^*$. 
Therefore its maximum is obtained when $\frac{d}{dz^*} \Delta(z^*) = 0.$ 
By basic calculus, it follows the maximizer $z_{\max}$ for Equation \eqref{eq:sn-curv}
is given by:

\begin{equation} 
\label{eq:zmax}
z_{\max} = \arg\max_{z^* \in (0, 1)} \Delta(z^*)  =  \frac{1}{\ln\left(1 + 1/\gamma \right)} - \gamma.
\end{equation}
Thus $\alpha_i(z)$ can be then be expressed in a
closed form by plugging in $z_{\max}$ from Equation \eqref{eq:zmax}
in for $z^*$, which can be simplified as:

\begin{equation*} 
\alpha_i(z) = \eta_i\left[\ln\left(\frac{1}{\gamma \ln(1+1/\gamma)}\right) + \gamma \ln\left(1+\frac{1}{\gamma}\right) - 1\right].
\end{equation*}

It can be verified that the derivative of $\alpha_i(z)$ with respect to $\gamma$ is negative for all $\gamma > 0$. Therefore since $\gamma = \omega + z$, 
the derivative of $\alpha_i(z)$ is also negative for all $z \geq 0$, so $\alpha_i(z)$ is maximized at $z = 0$. Since $\gamma = \omega$ when $z= 0$, the lemma follows.  
\end{proof}

We can now prove Theorem \ref{thm:snsw}.
\begin{proof}[Proof of Theorem~\ref{thm:snsw}]

By Theorem \ref{thm:add-alg-bound}, the algorithm achieves the desired run-time bound, since $v_i'(0) = \eta_i/\omega \leq 1/\omega$ (recall we normalized agent weights to be $\eta_i/\eta$) and the update on Line 8 in WBB takes $O(1)$ time.
Thus, we are left with bounding the approximation ratio of the algorithm for the product objective.  

Let $\POPT$ and $\LOPT$ denote the objective value of the optimal solution for the product and log objective, respectively.  
Consider the dual variables $(\beta_j, t_i)$ corresponding to the allocation returned by the algorithm.
Also by Theorem \ref{thm:add-alg-bound}, $\beta_j + \epsilon/m$ is a feasible dual solution to the dual program {\sc ASN-D}. 
Thus by Lemma \ref{lem:duality} we have:

\begin{align}
    \text{{\sc ASN-D}}(t,\beta) & = \sum_i \eta_i \left(\ln(t_i+ \omega) - \frac{t_i}{t_i+\omega}\right)+ \sum_j \left(\beta_j + \frac{\epsilon}{m}\right) \notag  \\
    &= \sum_i D(u_i) + \epsilon \geq \LOPT \label{eq:sn-dual-bound}. 
\end{align}

When the algorithm terminates we have $\eta_i\ln(u_i + \omega) \geq D(u_i)  - \alpha_i$ for every agent $i$. Along with Inequality \eqref{eq:sn-dual-bound}, 
this implies the total objective of the algorithm is bounded by: 

\begin{equation*}
\sum_i \eta_i \ln(u_i+\omega) \geq \sum_i (D(u_i) - \alpha_i) \geq \LOPT  - \sum_i \alpha_i - \epsilon.
\end{equation*}
From this inequality, and the  fact that $\LOPT = \ln\left(\POPT\right)$ (when weights are scaled such $\eta = 1$), it follows the algorithm's objective on the product objective is bounded by: 
\begin{alignat}{2}
\prod_i (u_i + \omega)^{\eta_i} = 
\exp\left(\sum_i \ln(u_i + \omega)\right) & \geq \exp\left(\LOPT - \sum_i \alpha_i -\epsilon  \right) \notag \\
& = \exp\left(- \sum_i \alpha_i -\epsilon  \right)\POPT. \label{eq:nsw-final}
\end{alignat}
From Lemma \ref{lem:nsw-diff-gap}, we have that $\sum_i \alpha_i  = O\left(\ln\left(\frac{1}{\omega\ln(1+\omega)}\right)\right)$, it follows $\exp\left(\sum_i \alpha_i +\epsilon\right) = O(e^{\epsilon}/(\omega\ln(1+1/\omega)))$. Therefore by rearranging the above inequality \eqref{eq:nsw-final}, the theorem is established. 
\end{proof}

\subsubsection{Extension to Smooth Piecewise-Linear Valuation}
\label{subsub:anw-pl}
The results in this section can be also extended to  piecewise-linear valuation, i.e., maximizing $\left(\prod_i (f_i(u_i + \omega))^{\eta_i}\right)^{1/\eta}$ where $f_i(\cdot)$ is a piecewise-linear function.  In the next section (Section \ref{subsec:app-pl-welfare}) we argue the multiplicative curvature $\mu_i$ is at most $4/3$ for all such functions. Combining arguments in this section  with that of Lemma \ref{lem:nsw-diff-gap} we can obtain the following result.

\begin{theorem} 
\label{thm:anw-pl}
Consider an ICA instance with $v_i(u_i) = \eta_i \ln(f_i(u_i + \omega))$ where $f_i(u_i)$ is a piecewise-linear function concave-additive function with $\min_{k}\left(x_{i,k+1} - x_{i,k}\right) \geq \max_{j} u_{ij}$.
If $\omega = \Omega(1)$, then there exists a polynomial-time algorithm for the problem with an $O(1)$ approximation factor.
\end{theorem}

The extension works since Lemma \ref{lem:nsw-diff-gap} can be adapted to derive an additive curvature $\alpha_i$ that is a constant if $\omega = \Omega(1)$.
For example, when $f_i(u_i) = \min(u_i, c_i)$ is a budget-additive function, we obtain an approximation of $\approx  1.154$ as $\epsilon \rightarrow 0$ and $\omega = 1$.
We defer the proof of this result to the full version of the paper.


\subsection{Utilitarian Welfare Maximization for Piecewise-Linear Concave Utilities}
\label{subsec:app-pl-welfare}


In this section we establish Theorem \ref{thm:pl-mult}, we analyze our multiplicative algorithm for $\ICA$ in the  special case where every valuation function $v_i(\cdot)$ is monotone piecewise-linear concave function. 
In this setting, each $v_i(\cdot)$ is defined over series of conjoined {\em segments} $\ell_{i,1}, \ell_{i,2}, \ldots, \ell_{i, \lambda_i}$. Let
\[
0 = x_{i,0} < x_{i,1} < x_{i,2} < \cdots < x_{i,\lambda_i -1}
\]
denote the transition points between the segments of $v_i(\cdot)$, where $x_{i,k}$ denotes the transition point between segments $\ell_{k}$ and $\ell_{k+1}$ along the $x$-axis. Also, let 
\[
v'_{i,1} > v'_{i,2} > \cdots > v'_{i,\lambda_i } \geq 0
\]
denote the slopes of the segments, where $v'_{i,k}$ denotes the slope of segment $\ell_{i,k}$ (the inequalities follow since $v_i(\cdot)$ is monotone and concave).
Thus, if $u_i \in [x_{i,k}, x_{i,k+1}]$, then we can express agent $i$'s valuation formally as $v_i(u_i) = v_i(x_{i,k}) + v'_{i,k}\cdot(u_i - x_{i,k})$.
We further assume that $\min_{k \in [\lambda_i-1]}\left(x_{i,k+1} - x_{i,k}\right) \geq \max_{j} u_{ij}$, i.e., the maximum additive utility earned from any one item is at most the length of any segment. Recall that for the special case of $v_i(u_i) = \min(u_i, c_i)$ for some budget $c_i$, this corresponds to the standard assumption that every agent $i$'s utility for each item is at most her budget.


In the definition of Algorithm \ref{alg:pd} in Section \ref{sec:ica-algo}, we assumed that every valuation $v_i$ is differentiable. Although a piecewise-linear $v_i(u_i)$ as defined above is not differentiable at each transition point $x_{i,k}$, Algorithm \ref{alg:pd} can be easily adapted to this setting by using supergradients in place of of tangents (we defer a complete definition of this adaptation to the full version).

Thus to establish Theorem \ref{thm:pl-mult}, it suffices to show the local multiplicative curvature $\mu_i$ for all such piecewise-linear functions is at most $4/3$. We first prove that the local multiplicative curvature of the budget-additive valuation $v_i(u_i) = \min(u_i, c_i)$ exactly $4/3$.
We then show the relaxing $v_i$ to be a general piecewise-linear function does not increase its value of $\mu_i$.


\begin{lemma} \label{fact:budget-3/4}
The local multiplicative curvature $\mu_i$ of $v_i(u_i) = \min(u_i,c_i)$ is exactly $4/3$.
\end{lemma}

\begin{proof}
For a fixed value $z$, let $\Delta_z(z^*, u_{i,j})$ be a function of $z^*$ defined by the expression inside the $\max$ in Equation \eqref{eq:mult-lcb} for $\mu_i(z, u_{i,j})$. 
For all $z+z^* \leq c_i$, $\Delta_z(z^*, u_{i,j})$ is given by:
\begin{equation} \label{eq:budget-gap}
\Delta_z(z^*, u_{i,j})  = \frac{z+z^*}{z+ z^*\left(\frac{u_{i,j}-z}{u_{i,j}}\right)}
\end{equation}

One can verify the derivative of $\Delta_z(z^*, u_{i,j})$ with respect to $z^*$ is positive for $z^* + z \leq c_i$, and also that the derivative of the expression for $\Delta(z^*, u_{i,j})$ when $z + z^* \geq c_i$ is negative. Thus it follows for a fixed $z$, $\Delta_z(z^*, u_{i,j})$ is maximized when $z+z^* = c_i$. 
Therefore, 
we can express $\mu_i(z, u_{i,j})$ as the following function of $z$:
\begin{equation} \label{eq:budget-gap-2}
\mu_i(z,u_{i,j})   = \frac{c_i}{z+ (c_i - z)\left(\frac{u_{i,j}-z}{u_{i,j}}\right)}
\end{equation}

The above expression is maximized when $z = c_i/2$. The expression is also strictly increasing in $u_{i,j}$; therefore by our assumption that $\max_{j} u_{ij} \leq c_i$, the expression is maximized with respect to $u_{i,j}$ when $u_{i,j} = c_i$. 
Thus plugging in $z = c_i/2$ and $u_{i,j} = c_i$, we obtain $\max \mu_i(z,u_{i,j}) = c_i/(c_i/2 + c_i/4) = 4/3$, as desired. 
\end{proof}

\begin{lemma}
If $v_i(\cdot)$ is piecewise-linear, concave, non-decreasing, and satisfies $\min_{k}\left(x_{i,k+1} - x_{i,k}\right) \geq \max_{j} u_{ij}$, then its multiplicative curvature $\mu_i$ is at most $4/3$.
\end{lemma}

\begin{proof}
Let $w = \max_j u_{i,j}$, let $(z,z^*)$ be the maximizers that define $\mu_i$, and let 
$\sigma = \sigma_i(z,w) = v_i(z + w) - v_i(z))/w$ denote the slope of the lower-bounding 
secant line given by Equation \eqref{eq:lcb-slope}.
Thus the equation of the lower bounding secant line $s(x)$ that defines $\mu_i$ is given by 
$s(x) = \sigma \cdot(x-z) + v_i(z)$. Let  $\ell_1, \ell_2$ denote the  segments containing $z$ and $z+w$ fall in, respectively.
By our assumption $\min_{k}\left(x_{i,k+1} - x_{i,k}\right) \geq \max_{j} u_{ij}$, $\ell_1$ and $\ell_2$ must be adjacent segments of the function.


To show the lemma, it suffices to show to show that for all $x \in [z, z+w]$, $v_i(x) / s(x) \leq 4/3$. Similar to proof of Lemma \ref{fact:budget-3/4}, we establish this by showing the largest gap occurs at the intersection of $\ell_1$ and $\ell_2$. Let $\ell_1, \ell_2$ be defined as functions of $x$ by $\ell_1(x) = v'_1x+b_1$ and $\ell_2(x) = v'_2x+b_2$, respectively.
Let $\overline{x}$ denote the $x$ coordinate of their intersection point, i.e., $\ell_1(\overline{x}) = \ell_2(\overline{x})$. Since $v_i(\cdot)$ is concave, we have $v'_1 < \sigma < v'_2$ and $b_1 < v_i(z) - zx < b_2$. For $x \in [z, \overline{x}]$, the multiplicative curvature at $x$ is the ratio between $\ell_1(x)$ and $s(x)$, which is
\[
\Delta(x) = \frac{v'_1x+b_1}{\sigma \cdot (x-z) + v_i(z)}.
\]
Since the derivative of $\Delta(x)$ is positive with respect to $x$, $\Delta(x)$ increases as $x$ increases.
We can similarly show that for $x \in [\overline{x}, z+z^*]$, $\Delta(x)$ decreases as $x$ increases. Thus, $\Delta(x)$ is maximized at $x = \overline{x}$.

To complete the argument, we can transform $\ell_2(x)$ so that conjoining segments $\ell_1$ and $\ell_2$ resemble a budget-additive function. In particular, we can ``flatten'' $\ell_2(x)$ by decreasing its slope $v'_2 = 0$, which  decreases the value of $\ell_2(z+z^*)$ from $v_i(\overline{x}) + v'_2(w-\overline{x})$ to now be $v_i(\overline{x})$. Since $z$ remains fixed, the slope of $s(x)$ decreases when adjust for this change, which means the curvature (evaluated at $\overline{x}$) can only increase. Applying Lemma \ref{fact:budget-3/4}, it follows that $\mu_i \leq 4/3$.



\end{proof}

\bibliographystyle{plainurl}
\bibliography{bibliography}


\end{document}